\documentclass[12pt,a4paper]{article}

\usepackage[left=1in,right=1in,top=1in,bottom=1in]{geometry}
\usepackage[table]{xcolor}

\usepackage{mathtools}
\usepackage[table]{xcolor}
\usepackage{multirow}

\usepackage{amsmath,amsthm,amssymb,amsfonts,setspace}
\usepackage{bbm}
\usepackage{graphicx}
\usepackage{rotating}
\usepackage{enumitem}
\usepackage{eurosym}
\usepackage{epstopdf}
\usepackage{csquotes}
\usepackage[natbibapa,nodoi]{apacite}
\usepackage{array}
\usepackage{caption}
\usepackage{hyperref}
\usepackage{collref}
\usepackage{grffile}
\usepackage{subcaption}
\usepackage{xcolor}
\usepackage{mathrsfs}  
\usepackage{verbatim}
\usepackage{bbm}
\usepackage{tikz}
\usetikzlibrary{decorations.pathreplacing,calligraphy,decorations.markings}
\usepackage{calc}
\usepackage{array}
\usepackage{pgfplots}
\usepackage{multirow}
\usepgfplotslibrary{fillbetween}
 
\DeclareMathOperator*{\argmin}{arg\,min}

\usepackage{lineno}

\newtheorem{prp}{Proposition}
\newtheorem{theorem}{Theorem}

\newenvironment{example}{
    \begin{trivlist}
    \item[\hskip \labelsep{\bfseries Example:}]}{\end{trivlist}}

\newcommand{\CP}{\mathcal{C}\mathcal{P}}
\newcommand{\CV}{\mathcal{C}\mathcal{V}}

\newcommand{\cdf}[1][]
{
	\ifthenelse{\isempty{#1}}{
	\ensuremath{F\xspace}}{\ensuremath{F(#1)\xspace}}
} 

\newcommand{\pdf}[1][]{
	\ifthenelse{\isempty{#1}}{
	\ensuremath{f\xspace}}{\ensuremath{f(#1)\xspace}}
} 

\usepackage{cleveref}

\newtheorem{lemma}{Lemma}

\newtheorem{proposition}{Proposition}

\begin{document}
\title{
  Propose or Vote: A Canonical Democratic Procedure\footnote{We thank Pio Blieske, Barton Lee, Giorgis Giovanni, Lukas Kolleck and Cesar Martinelli and participants at the 8th ETH/CEPR Workshop on Democracy for valuable comments and suggestions.} 
}


\author{\normalsize Hans Gersbach\\
	\small KOF Swiss Economic Institute \\[-1.6mm]
	\small ETH Zurich and CEPR \\[-1.6mm]
	\small Leonhardstrasse 21 \\[-1.6mm]
	\small 8092 Zurich, Switzerland \\[-1.6mm]
	\small hgersbach@ethz.ch\\
}

\date{
	This version: Januar 2026
}

\maketitle
\begin{center}
\end{center}

\vspace*{-1.7cm}

\vspace{-1cm}
\begin{abstract}

\noindent


\singlespacing

This paper introduces Propose or Vote (PoV), a democratic procedure for collective decision-making and elections that does not rely on a central mechanism designer. In the first stage, members of a polity choose whether to become proposal-makers or to participate only as voters. In the second stage, voters decide by majority voting over the set of submitted proposals. With appropriately chosen default points, PoV implements the Condorcet winner in a single round of voting whenever one exists. We show that this implementation is globally unique when the number of members is odd; for an even number of members, uniqueness can be restored by adding an artificial agent. PoV can also be applied to elections, where agents decide whether to stand as candidates or vote over the resulting candidate set.

\medskip


\medskip


\medskip

\end{abstract}

\newpage

\clearpage
\pagenumbering{arabic} 


\section{Introduction}\label{sec:introduction}

In various collective decision-making settings — such as committee or polity deliberations and legislative processes — participants must select a single alternative from a vast array of options. A common approach in these contexts is to employ successive rounds of pairwise majority voting. For example, legislative bodies refine bills through sequential votes on amendments, while negotiating parties in multi-agent resource allocation problems often reach consensus through step-by-step adjustments. However, when the number of alternatives is large, it becomes impractical to evaluate all options exhaustively. Furthermore, the final outcome is highly contingent on both the subset of alternatives selected for voting and the order in which they are considered. This creates a well-documented vulnerability: strategic agenda setters can manipulate the decision-making process to steer outcomes in their favor, a phenomenon extensively studied in the 
literature.\footnote{For foundational contributions, see \cite{McKelvey1979}, \cite{Rubinstein1979}, \cite{Bell1981}, and \cite{Schofield1983}. More recent research has explored the scope and limitations of agenda power in various contexts, including \cite{Barbera2017}, \cite{Nakamura2020}, and \cite{Nageeb2023}, \cite{Horan2021} and \cite{Rosenthal2022}, while \cite{Banks2002} provides a review of earlier work.}

In this paper, we introduce a procedure to address this challenge. 
In the initial stage, all members of the polity choose whether to participate as proposal-makers or to reserve their role for voting on proposals in the subsequent stage. This process is called Propose or Vote (PoV). 

The PoV procedure governs proposal-making through the following set of rules: If no agent chooses to make a proposal, an alternative is randomly selected as the society’s final choice.\footnote{An alternative would be to introduce a status quo that is implemented in this case. However, we do not pursue this here.}
If a single agent opts to propose, his/her proposal is automatically implemented. If two agents choose to propose, their proposals are put to a vote among the remaining members of society. If more than two agents opt for proposal-making, two proposals from the group are randomly selected to be put to a vote. If everyone opts for making a proposal, again an alternative is randomly selected as the society’s final choice. 

We examine whether this idea works in the simplest context with a one-dimensional set of alternatives and a society of agents with symmetric, single-peaked preferences and a publicly known distribution. We analyze the strategic incentives created by the PoV procedure for agents, and examine whether it works in selecting the socially optimal Condorcet winner in one round of voting.

We show that with an odd number of agents, there exists an equilibrium in which precisely the median agent applies for proposal-making, proposes his or her peak, and the peak is implemented. When the number of agents is even, both the left median and the right median opt for proposal-making, propose their peaks, and one of them will be selected by the remaining agents.  

We then explore the uniqueness of the result. With an even number of agents, we show with counterexamples that there are other equilibria in which more extreme proposals are implemented. With an odd number of agents, we show uniqueness. One can always ensure that there is an odd number of agents by adding an artificial voter (AI voter) when necessary.  
Moreover, we show that the procedure can be used in the same way for elections. Agents can decide to propose themselves (or somebody else) for an office, and the agents who have not made a proposal elect one candidate from the set of candidates.     

The paper is organized as follows. In the next section, we relate our work to the literature. In Section \ref{model}, we introduce the model, and in Section \ref{Simple-Result}, we present the equilibria in which the Condorcet winner is implemented. In Section \ref{Ohter-Equilibria}, we look at other equilibria. In Section \ref{GlobalUniqueness}, we examine how to achieve global uniqueness. In Section \ref{elections}, we explore the application to elections. Section 
\ref{section:conclusion} concludes. Longer proofs are relegated to the Appendix. 

\section{Relation to Literature}

Following the seminal contribution by \cite{Moulin1980}, it is well-established that with single-peaked preferences, the Condorcet winner can be implemented by asking agents to report their peak, introducing a set of fixed peaks, and applying a generalized median rule. This foundational insight has been significantly generalized by \cite{Border1983}, \cite{Barbera1994}, and \cite{Klaus2020}. In contrast to their mechanism design perspective, our paper addresses a different implementation challenge: the absence of a centralized mechanism designer.  In many institutional procedures for committee or parliamentary decisions, procedures rely on agent-driven proposal-making and voting, as a centralized mechanism designer is absent. In this paper, we explore how a PoV procedure can solve this problem.\footnote{As we focus on a one-dimensional policy space, our work is related to the spatial bargaining literature, including \cite{Baron1991}, \cite{Banks2000}, \cite{Kalandrakis2016}, and \cite{Zapal2016}, and a single payoff-relevant outcome. \cite{gersbach2025votingrandomproposersrounds} demonstrated that the random selection of agenda-setters in an iterative majority voting framework may provide them with incentives to propose the Condorcet winner.}

Our paper is also related to the algorithmically driven and iterative voting procedures with random elements in various policy spaces, as developed by \cite{Airiau2009}, \cite{Goel2012}, and \cite{Garg2019}, among others. We contribute to this literature by analyzing a simple procedure for implementing the Condorcet winner. 

We are not aware of any real-world procedure that features the characteristics of PoV. Yet, there are parallels. In some modern democratic parliaments, there are members who have the right to introduce
legislation or motions but do not have voting rights. This arrangement exists for various constitutional, procedural, or representational reasons. Examples include the Speaker of the House of Commons in the UK\footnote{The Speaker may propose procedural motions but does not vote unless there is a tie, see Erskine May's Parliamentary Practice; House of Commons Factsheet M2.} or Non-voting Delegates in the United States Congress\footnote{Delegates from U.S. territories (e.g., Puerto Rico, Guam) can introduce bills and vote in
committees but cannot vote on the final passage of legislation on the House floor (see U.S. House Rules, Rule III §3)}. Parliamentary Secretaries in Canada can introduce government business but cannot vote in committees\footnote{ House of Commons Procedure and Practice, 3rd ed.; Parliament of Canada Act, Sections 50-51.}.

In some sense, the Speaker of the House of Commons in the UK acts like the AI voter in our context, as he/she only votes if there is a tie. This example illustrates that separating proposal rights from voting power is used in practice. 
Finally, the procedure may be sustained by social norms. For instance, if someone proposes themselves as a candidate for office, one often observes that this person does not participate in voting on the set of candidates put forward.

\section{The Model}
\label{model}
Consider a society that consists of $N$ agents. They have to select an alternative $x$ from a one-dimensional policy space, which we take as an interval of arbitrarily chosen but fixed size $[-A,A]$ for some $A \in (0,\infty)$.

To do so, we introduce the following procedure, called Propose or Vote (PoV), which is structured into two stages.

\begin{enumerate}
    \item Each agent decides whether he will propose a policy $x_p \in [-A,A]$ or be placed into the pool of voters. The set of proposers will be denoted by $\CP$ and the set of voters by $\CV$.

    \item Among all those proposals, two are selected uniformly at random and put to a vote. Only the agents who opted to be placed in the voter pool $\CV$ in the first stage are allowed to participate in the voting procedure. The winner of the vote according to the majority rule is implemented.
\end{enumerate}


Each of the society's agents $i \in [N] = \{1, ... , N\}$ is  characterized by his/her type $\theta_i\in [-A,A]$ and the types satisfy $\theta_1 < \dots < \theta_i < \theta_{i+1} < \dots < \theta_N$. The utility of an agent of type $\theta_i$ from an alternative $x \in[-A,A]$ is given by 
$$u_{\theta_i,x}:=-(x-\theta_i)^2.$$ 

Therefore, it has a unique maximum at $\theta_i$ and is symmetric around $\theta_i$. 
We will use the notation $u_{\theta_i, v}$ to represent the expected utility of an agent who chooses to vote. 

Let $M$ be the median of $[N]$: if $N$ is even, $M$ is the average of the two middle elements of $[N]$; if $N$ is odd, $M$ is the unique middle element of $[N]$. In the case of an odd number of agents, we denote the unique median agent as $m:=M$, which, due to the ordering of the $\theta_i$, is precisely the median of the set $[N]$. If there is an even number of agents, there might not be a unique median agent. Then, we look at the agents to the left and right of $M$, denoting the left median agent as $m_l := \max\{i \in [N] \mid i \leq M\}$ and the right median agent as $m_r := \min\{i \in [N] \mid i \geq M\}$.

In the PoV-procedure described in two stages above, some edge cases might occur that we have not dealt with yet. We impose the following rules to resolve these: 
 \begin{enumerate}
 \item If no agent chooses to make a proposal, an alternative is selected uniformly at random from the agents' types $\{\theta_i \in [-A,A]\mid i \in [N]\}$ and is implemented without any vote.
 \item If a single agent opts to propose, their proposal is automatically implemented.
 
 \item If two or more agents choose to propose, the general procedure described above applies.
 
 \item If the vote ends in a draw, a coin flip decides which of the two selected proposals is implemented. In particular, if everyone opts to propose, two of the submitted proposals are selected uniformly at random, where one chooses an $a$ uniformly at random from $\CP$ and then a $b$ uniformly at random from $\CP \setminus \{a\}$. Subsequently, a coin flip decides between the two. This is the same as implementing one of the proposals uniformly at random.

\end{enumerate}

To simplify the analysis, we impose the following tie-break rules. 

\begin{enumerate}

\item If an agent expects the same utility from both voting and submitting a proposal, he/she decides to vote.

\item If a voter is indifferent between two proposals, he/she abstains from voting.

\end{enumerate}

The first tie-breaking rule is quite important. In order to ensure that proposal-making is not too attractive, one could impose an arbitrary, sufficiently small, and positive cost for proposal-making. This would strengthen the results, as agents will not make a proposal if the proposal does not improve the agent's expected utility. Notice that the cost of the proposal also needs to be sufficiently small to guaranty that for those agents for whom proposal-making was strictly better than voting, this remains true. 

To analyze the dynamic voting game, we use the equilibrium concept of Subgame Perfect Nash Equilibrium (SPNE) with the refinement that agents in the last stage eliminate weakly dominated strategies. Since there is only one voting round, this essentially means that all agents who have retained their voting rights will vote sincerely in the final stage and thus vote for the alternative they prefer more than the other one, if two alternatives are put to a vote.

 \section{The Simple Result}
\label{Simple-Result}

In this section, we analyze how proposers choose their optimal proposals, taking into account their types, and the resulting equilibria of the PoV procedure.
The following theorem shows the existence of equilibria in the PoV-procedure.

The following theorem shows the existence of equilibria in the PoV-procedure.

\begin{theorem}\label{theorem:vop-theorem}
Suppose that we have a polity with $N$ agents, each characterized by their type $\theta_i \in [-A,A]$. Then there exists an equilibrium for the Propose or Vote-procedure among the $N$ agents. This equilibrium is given as follows:
\begin{enumerate}
    \item If $N$ is odd, the constellation where the median agent $m$ of type $\theta_m$ opts for submitting the proposal $\theta_m$ and all other agents opt for taking the vote is an equilibrium.
    \item If $N$ is even, the constellation where the left and the right median opt for submitting their own type as a proposal and all the other agents opt to vote is an equilibrium.
\end{enumerate}
\end{theorem}
The proof of Theorem \ref{theorem:vop-theorem} is given in the appendix. 
Theorem \ref{theorem:vop-theorem} illustrates the inherent trade-offs. Consider, for example, the case where $N$ is odd. Agents other than $m$ face the following dilemma: by choosing to become proposers, they lose their voting power, which may result in the final decision either remaining unchanged or shifting further from their preferred alternative. As a result, they have no incentive to take on the role of proposer. If $N$ is even, agents other than the two median types have no incentive to become proposal makers. Suppose that the agent is to the left of the left median. Then, becoming a proposal-maker and offering his/her peak would shift the probability that the right median is selected to $\frac{2}{3}$ and the probability that the left median is selected is equal to $\frac{1}{3}$. If the agent remains a voter, the probability of the left median being selected is $\frac{1}{2}$, and the probability that the right median is selected declines to $\frac{1}{2}$.

\section{Other Equilibria}
\label{Ohter-Equilibria}

We next discuss whether other equilibria exist. 

\subsection{Results for $N$ odd}

When $N$ is odd, we can sharpen the main result from the last section. 

\begin{prp}\label{Proposition:Only-one}
    Suppose that $N$ is odd and an equilibrium has only one proposer. Then, this has to be the equilibrium found in Theorem \ref{theorem:vop-theorem}.
\end{prp}

The intuition for this result is simple. Suppose one agent other than $m$ is proposing in equilibrium. Since no one else is proposing, he will propose his type. Then, the median is in the group of voters. If he proposes his type, he will surely win the election, as being the median will guarantee him at least half of the voters. This will make him better off. This proves that no equilibrium with only one proposer exists if the proposer is not the median agent.

In addition, we can show that there does not exist an equilibrium with exactly two proposers when $N$ is odd.

\begin{prp}\label{proposition:no-two-proposers-N-odd}
    Assume that $N$ is odd. Then, there are no equilibria with $|\mathcal{CP}|=2$.
\end{prp}

The logic here is more subtle. One has to first establish which proposals would be made in equilibrium. In the Appendix, we show that  if the number of voters $|\mathcal{CV}|$ is odd, which would happen if two agents opt for proposal-making, it is enough to consider the equilibria where for all $p \in \mathcal{CP}$ we have $x_p \in [\min\{\theta_p, \theta_{mv}\}, \max \{\theta_p, \theta_{mv}\}]$, where $\theta_{mv}$ is the type of the median among the voters. In any of these possible equilibria, all proposers would propose their own peak. Again, we can prove that there are always incentives to deviate from a constellation with two proposers. The detailed proof is provided in the Appendix.

\subsection{Counterexample for N even}

We now provide a counterexample to prove that, for an even number of agents, there can be equilibria (other than the one characterized in Theorem \ref{theorem:vop-theorem}) with exactly two proposers. Indeed, considering the case with $N=4$ agents, we provide an example of an equilibrium that differs from the one in Theorem \ref{theorem:vop-theorem} and that has exactly two proposers.

\begin{example}[Counterexample]
    Let $N=4$ and let the types be given by $\boldsymbol{\theta}=(-4, -3, 3, 4)$. Suppose that the two most extreme agents, those with type $\theta_1$ and $\theta_4$, decide to propose their own peaks, and the other two, those with type $\theta_2$ and $\theta_3$, are voting. We claim that this is an equilibrium. 
    
    \begin{minipage}{\linewidth}
        \vspace{10pt}
        \centering
        \begin{tikzpicture}
            \draw (-5,0) -- (5,0);
            \draw (-4,-0.2) -- (-4, 0.2); \node at (-4.5, -0.5) {$x_{1}=\theta_1$};
            \draw (-3,-0.2) -- (-3,0.2); \node at (-3, -0.5) {$\theta_2$};
            \draw (3,-0.2) -- (3,0.2); \node at (3, -0.5) {$\theta_3$};
            \draw (4,-0.2) -- (4,0.2); \node at (4.5, -0.5) {$\theta_4=x_{4}$};
        \end{tikzpicture}
        
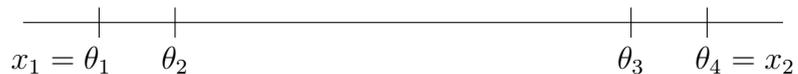
\captionof{figure}{Alternative equilibrium where the two most extreme agents propose their own peaks.}
        \label{fig:no_uniqueness_n=4}
        \vspace{15pt}
    \end{minipage}

    Observe that the left proposal, $x_1=\theta_1$, is closer to the left voter's type $\theta_2$ and the right proposal, $x_4=\theta_4$, is closer to the right voter's type $\theta_3$. Hence, the vote ends in a draw, and each of them is implemented with probability $1/2$. Consider the left proposer. Proposing its own peak leads to an expected utility of
    \begin{equation*}
        \mathbb{E}[u_{\theta_1, \theta_1}]= -\frac{1}{2}(\theta_4-\theta_1)^2=-32.
    \end{equation*}
    If agent $1$ decides to vote instead, the other proposal would win the vote and therefore lead to an expected utility of 
    $$ \mathbb{E}[u_{\theta_1, v}] = -(\theta_1-\theta_4)^2=-64.$$
    If agent $1$ decides to move its proposal to the right, the new proposal must win against the other proposal $x_4$, since otherwise it only reduces the expected utility. Hence, the closest possible proposal that has a positive probability of winning against $x_4$ is $\tilde{x}_1=\theta_3-(x_4-\theta_3)=2 \theta_3 - \theta_4= 2$. However, if $\tilde{x}_1$ is implemented, agent $1$ gets a utility of $-(\tilde{x}_1- \theta_1)^2=-36$. Thus, there exists no profitable deviation for the left proposer.

    Next, we consider the left voter of type $\theta_2$ who obtains under the current strategy an expected utility of
    \begin{equation*}
        \mathbb{E}[u_{\theta_2, v}]=-\frac{1}{2}\left((\theta_1-\theta_2)^2+(\theta_4-\theta_2)^2 \right)=-25.
    \end{equation*}
    
    If agent $2$ decides to make a proposal instead of voting, then the new median voter is $3$. If agent $2$ makes a proposal $\tilde{x}_2$ that is strictly farther from $\theta_3$ than $x_4=\theta_4$, $x_4$ will be implemented with a probability of $2/3$. Thus,
    \begin{equation*}
        \mathbb{E}[u_{\theta_2, \tilde{x}_2}]\leq - \frac{2}{3} (x_4-\theta_2)^2=-\frac{98}{3}<-25.
    \end{equation*}
    If $\tilde{x}_2$ wins against $x_4$, it has to hold that $\tilde{x}_2>\theta_3-(x_4-\theta_3)=2\theta_3-\theta_4=2$. But then
    \begin{equation*}
        \mathbb{E}[u_{\theta_2, \tilde{x}_2}] = - \frac{2}{3}  (\tilde{x}_2 - \theta_2)^2 - \frac{1}{3} (x_4- \theta_2)^2  \leq -\frac{99}{3} < -25.
    \end{equation*}
    Hence, there does not exist a suitable deviation from the current strategy for agent $2$.
    
    Due to symmetry, we can show that there exists no profitable deviation for players of types $\theta_3$ and $\theta_4$ by using the same arguments as for $\theta_2$ and $\theta_1$, respectively. Thus, we can conclude that we have indeed found an equilibrium.

\end{example}

\section{Global Uniqueness}\label{GlobalUniqueness}

\subsection{Global Uniqueness with $N$ Odd}
We now establish the global uniqueness of the equilibrium characterized in Theorem \ref{theorem:vop-theorem} when $N$ is odd. 

\begin{theorem}\label{theorem:Uniqueness_of_N_odd}
    Given an odd number of agents $N$, there exists no equilibrium other than the one described in Theorem \ref{theorem:vop-theorem}.
\end{theorem}

The proof of Theorem \ref{theorem:Uniqueness_of_N_odd} is provided in the Appendix. This theorem constitutes the second main result of our paper. The intuition behind it is as follows: 

When there is only one agent proposing, we have shown that no other agent has an incentive either to propose or to vote. We then consider all other possible cases in which the number of proposers differs from one. In each such case, there always exists at least one agent who is proposing but has an incentive to switch to voting. The key idea is that, with an odd number of agents, one can always identify a proposal that would lose if put to a vote. For the agents behind such a proposal, voting becomes strictly preferable to proposing. The full, detailed proof is presented in the Appendix.

\subsection{Global Uniqueness with N even}
As we have seen in the previous section, when the number of agents is even, the equilibrium is generally not unique. To ensure that the number of agents is always odd, we introduce a preliminary stage in which an artificial agent is added to the system. This artificial agent, referred to as the \emph{AI voter}, participates in the voting process only when the number of real agents is even. The AI voter can be thought of as a machine that casts an additional vote. The mechanism of this additional agent is as follows: when introducing it, a fair coin is flipped, and this agent is assigned either peak $A$ or $-A$, depending on the outcome. From this point forward, it acts as a usual agent. Note that, in order to ensure that the strict ordering of the preferences still holds, one may have to slightly enlarge their possible support from $[-A,A]$ to $[-\tilde A,\tilde A]$ for some $\tilde A > A$. In this new model, there is again a unique equilibrium.

\begin{theorem}\label{thm:model_iii_uniqueness}
    Let the number of agents be even; the introduction of an AI voter will make it odd. The scenario in which the median agent $m$ of the altered set of agents (including the AI voter) opts to present the proposal $\theta_m$ and all other agents opt to vote is the unique equilibrium.
\end{theorem}
\begin{proof}
The proof of Theorem \ref{thm:model_iii_uniqueness} builds on the theory we have established for odd $N$. We obtain the existence of an equilibrium from Theorem \ref{theorem:vop-theorem} and its uniqueness readily from Theorem \ref{theorem:Uniqueness_of_N_odd}.
\end{proof}

\section{Elections}
\label{elections}

\subsection{Application}

The procedure can be applied in the same way to an election to office. Agents decide whether to apply for proposal-making or voting. Now, a proposal is a candidate for office. An agent opting for proposal-making may propose himself/herself for the office or some other agents. The remaining agents then vote. Our analysis and results can be readily transferred to such a setup. The only difference is that the set of possible alternatives is now constrained to the set of agents $N$ and is a discrete set of alternatives. Note that a candidate who is proposed by another agent can vote.\footnote{Again, one can add an AI voter to achieve an odd number of participants, but the AI voter is only active in voting. Yet, there are a couple of subtle issues with an artificial agent participating in elections, which we will spell out in future versions of the paper.} Our results generalize to elections. 

In the example that follows, we demonstrate how the procedure ensures a unique equilibrium when three agents are involved.

\subsection{Illustration for \( N = 3 \) Agents}
\label{IllustrationN3Agents}
Consider a society with three agents, where their types are given by \(\theta_1 < \theta_2 < \theta_3\). Each agent can choose to either propose a candidate or vote. We explore all possible constellations of proposers and voters in the Appendix and prove that there is a unique equilibrium. In the following, we simply report that with one proposer, there is a single equilibrium.

Suppose that only one agent opts to make a proposal. The proposer suggests a candidate, and the other two agents vote. The proposer will propose the candidate who maximizes their utility and, hence, themselves. 

\begin{table}[h]
\centering
\renewcommand{\arraystretch}{1.5}
\begin{tabular}{|c|c|c|c|}
\hline
Proposer & Utility for Agent 1 & Utility for Agent 2 & Utility for Agent 3 \\
\hline
Agent 1 & 0 & \(-(\theta_2 - \theta_1)^2\) & \(-(\theta_3 - \theta_1)^2\) \\
\hline
Agent 2 & \(-(\theta_1 - \theta_2)^2\) & 0 & \(-(\theta_3 - \theta_2)^2\) \\
\hline
Agent 3 & \(-(\theta_1 - \theta_3)^2\) & \(-(\theta_2 - \theta_3)^2\) & 0 \\
\hline
\rowcolor{gray!30}
Equilibrium & No & Yes & No \\
\hline
\end{tabular}
\caption{Utilities realized by each agent when only one agent proposes.}
\label{tab:one_proposer}
\end{table}
\begin{itemize}
    \item If agent 1 proposes himself (or, more precisely in the model's terms, his own peak), agent 2 has an incentive to propose himself, which would yield a higher utility of \(0\). Indeed, if agent 2 proposes, the proposals would be put to a vote, and since agent 3 would choose the proposal closer to his type, agent 2 would win the election since \(\theta_2\) is closer to \(\theta_3\) than \(\theta_1\).

    \item If agent 2 proposes himself, neither agent 1 nor agent 3 has an incentive to deviate. Even if they were to propose, they would not achieve a better outcome than by voting. Specifically, any proposal they make would either result in the same utility or a lower one compared to voting. Given the condition that agents choose to vote if they are indifferent between voting and proposing, this implies they have no incentive to propose.

    \item If agent 3 proposes \(\theta_3\), agent 2 has an incentive to propose \(\theta_2\), which would yield a higher utility of \(0\).
\end{itemize}

Thus, the only stable scenario is when the median agent (agent 2) proposes himself. This aligns exactly with the equilibrium characterized by Theorem \ref{theorem:vop-theorem}.

\section{Discussion and Outlook}
\label{section:conclusion}

We demonstrated that there exists an equilibrium in which the Propose-or-Vote mechanism efficiently implements the socially optimal alternative in a simple setting, and that the Propose-or-Vote (PoV) procedure selects the Condorcet winner in a single voting round.
We have used a simple framework in which the distribution of the peaks of individual preferences is public knowledge, and consequently the peaks are as well.

Several potential extensions suggest themselves for future research. These include environments with asymmetric information, richer distributions of agents and utility functions, and alternative voting procedures. For instance, one might allow for a run-off among all submitted proposals rather than restricting attention to two randomly selected ones. Another possibility is to introduce a preliminary stage in which the set of prospective proposers is revealed prior to proposal submission. In our framework, however, uniqueness in the implementation of the Condorcet winner implies that such modifications cannot enhance the outcome. On the contrary, they may generate additional and undesirable equilibria in the present environment.\footnote{Illustrative examples are available upon request.} Nevertheless, in alternative institutional or informational settings, such procedural refinements might well yield novel dynamics and welfare implications worthy of further investigation.

\newpage

\bibliographystyle{apacite}
\bibliography{vrp}
\newpage

\appendix

\section{Appendix} \label{sec_proofs}
\renewcommand{\theproposition}{\Alph{section}.\arabic{proposition}}
\setcounter{proposition}{0}

\renewcommand{\thelemma}{\Alph{section}.\arabic{lemma}}
\setcounter{lemma}{0}

\renewcommand{\thefigure}{A.\arabic{figure}}
\setcounter{figure}{0}

\subsection{Proof of Theorem \ref{theorem:vop-theorem}}
\begin{proof}
    We start with the case where $N$ is odd. We show that the situation in which only the median voter opts to propose his own type and all other agents choose to vote is an equilibrium; that is, no agent has the incentive to change his strategy in this situation, given that all other agents stick to their current choice. We now show that any agent other than the median agent $m$ will stay with opting to vote in the first stage. This follows from the fact that agent $i \neq m$ never improves his payoff by deviating from voting to proposing. By symmetry, it suffices to deal with the case that an agent $i < m$ deviates to proposing some $x_i \in [-A, A]$ instead of voting. After the deviation, we have $\widetilde {CP} = \{i, m\}$ and $\widetilde {CV} = [N] \setminus \{i, m\}$. The median voter in $\widetilde {CV}$ then has type $\theta_{m+1}$, and note that $\theta_m < \theta_{m+1}$.
    First assume that $x_i \leq \theta_m$. Both $x_i$ and $\theta_m$ lie to the left of $\theta_{m+1}$, and $\theta_m$ is to the right of $x_i$. Hence $\theta_m$ is closer to $\theta_{m+1}$ than $x_i$ is, so the majority in $\widetilde {CV}$ strictly prefers $\theta_m$. The implemented policy is $\theta_m$, and agent $i$'s payoff is unchanged. On the other hand, assume that $x_i > \theta_m$. Since
    \[
    \theta_i < \theta_m < x_i,
    \]
    we always have
    \[
    |x_i - \theta_i| = x_i - \theta_i > \theta_m - \theta_i = |\theta_m - \theta_i|,
    \]
    so whenever $x_i$ is implemented, agent $i$'s utility satisfies
    \[
    u_i(x) = -(\theta_i - x_i)^2 < -(\theta_i - \theta_m)^2 = u_i(\theta_m).
    \]
    If $x_i$ loses and $\theta_m$ is implemented, his utility is exactly $u_i(\theta_m)$. Thus, no agent with $i < m$ can profit by deviating.
    
    Furthermore, the median agent $m$ has no incentive to change from proposal-making to voting as the implemented policy would then be chosen uniformly at random from the agents’ types $\{ \theta_i \}_{i \in \{1, \ldots, N\}}$, leaving the median agent with a worse expected utility of 
    $$\mathbb{E}(u_{\theta_m,v}) = - \frac{1}{N}\sum_{i = 1}^N (\theta_m - \theta_i)^2 < 0 =u_{\theta_m,\theta_m}.$$ 
    Moreover, the median voter maximizes his utility by proposing his own type if he is the only proposal-maker; thus, the median voter also has no incentive to change his proposal. In the second stage, there is only one choice for the voters, so they cannot change their strategy in this last round. Hence, we found an equilibrium.

    Consider now the case where $N$ is even. We show that the situation in which both the left median agent $ m_l$ and the right median agent $m_r$ propose their own type and all the other agents vote is an equilibrium. Notice that $m_l \neq m_r$ since the agents' types are pairwise different and we face a situation with an even number of agents.
    
     In the second stage, voters will always vote for the policy that is closest to their type to maximize their utility. Hence, proposing their own type is the best strategy for the left median, as his utility of proposing his own type is $$\mathbb{E}(u_{\theta_{m_l}, \theta_{m_l}}) = -\frac{1}{2} (\theta_{m_l} - \theta_{m_r})^2 ,$$ while his utility from proposing $x_{m_l}$ would be $$\mathbb{E}(u_{\theta_{m_l}, x_{m_l}}) = -\frac{1}{2}\cdot (\theta_{m_l} - \theta_{m_r})^2 - \frac{1}{2}\cdot (\theta_{m_l} - x_{m_l})^2$$ if he changed his proposal only slightly enough that his proposal $x_{m_l}$ is still preferred by half of the voters. Else, his expected utility would become $$\mathbb{E}(u_{\theta_{m_l}, x_{m_l}}) = -(\theta_{m_l} - \theta_{m_r})^2 ,$$ since the right median's proposal would be implemented every time. In both cases, his expected utility decreased compared to proposing $\theta_{m_l}$. Hence, he has no incentive to change his proposal. For the left median's choice in the first stage, notice that his utility, if he switched to voting, would again be $$\mathbb{E}(u_{\theta_{m_l}, v}) = -(\theta_{m_l} - \theta_{m_r})^2 ,$$ as $\theta_{m_r}$ would always be implemented by being the only proposal made. Hence, the left median has no incentive to change his strategy. A similar argument applies to show that the right median has no incentive to change his strategy. Finally, the voters do not have an incentive to switch to proposal-making. 
    To show this, we first prove that no agent has an incentive to propose an alternative different from their own type. 

Consider an agent located to the left of the left median. Clearly, there is no incentive for such an agent to propose a point to the left of their type, as this would only reduce the proposal's appeal and lead to a lower expected utility. 

If the agent retains their voting power, their expected utility is
\[
\mathbb{E}(u_{\theta_i,v}) = -\frac{1}{2} \left((\theta_i - \theta_{m_l})^2 + (\theta_i - \theta_{m_r})^2\right).
\]
If the agent instead becomes a proposer and loses their vote, proposing, their expected utility becomes
\[
\mathbb{E}(u_{\theta_i, x_i}) = -\frac{1}{3} (\theta_i - \theta_{m_l})^2 - \frac{2}{3} (\theta_i - \theta_{m_r})^2.
\]
We now compare the two:
\[
\mathbb{E}(u_{\theta_i,v}) > \mathbb{E}(u_{\theta_i, x_i}) \quad \Leftrightarrow \quad \frac{1}{6} (\theta_i - \theta_{m_r})^2 > \frac{1}{6} (\theta_i - \theta_{m_l})^2 \quad \Leftrightarrow \quad |\theta_i - \theta_{m_r}| > |\theta_i - \theta_{m_l}|.
\]
This inequality holds by assumption, as we are considering an agent who is closer to the left median than to the right median. This proves that there is no incentive to propose for agents that are currently voting. A symmetric argument applies to agents located to the right of the median.

Altogether, we conclude that no agent has an incentive to deviate from their strategy. Hence, we have identified an equilibrium.
\end{proof}

\subsection{Proof of Proposition \ref{Proposition:Only-one}}
\begin{proof}
We assume that $N$ is odd. Suppose that there is an equilibrium with only one proposer $p$ while the remaining agents vote. We first prove that $x_p=\theta_p$. For this, assume by contradiction that $x_p \neq \theta_p$. However, the proposer can increase its utility by changing its proposal to $\widetilde{x}_p=\theta_p$. Since all other agents vote and no one else makes a proposal, the new proposal $\widetilde{x}_p$ is implemented. Thus, the proposer $p$ receives the highest utility of $0$. Next, we prove that the proposer $p$ has to be the median. Again, assume by contradiction that $p$ is not the median agent. Then, without loss of generality, assume $\theta_p < \theta_m$; in the other case, a symmetric argument applies. Now, the agent $p+1$, offering $\theta_{p+1}$, would win the voting and be better off than before; hence, it is not in equilibrium, leading to a contradiction. Hence, the only equilibrium with one proposer is the one in which the median agent proposes its own peak.
\end{proof}

\begin{figure}
\begin{center}
\begin{tikzpicture}[scale=1.5]
    \draw[->] (-3,0) -- (3,0) node[right]{$\theta$};

    \filldraw (0,0) circle (1.5pt) node[below]{$\theta_m$};
    \filldraw (-2,0) circle (1.5pt) node[below]{$\theta_p$};
    \filldraw (-1,0) circle (1.5pt) node[below]{$\theta_{p+1}$};
    \filldraw (-0.75,0) circle (1.5pt) node[below]{};
    \filldraw (-0.5,0) circle (1.5pt) node[below]{};
    \filldraw (1,0) circle (1.5pt) node[below]{};
    \filldraw (0.75,0) circle (1.5pt) node[below]{};
    \filldraw (2.5,0) circle (1.5pt) node[below]{};
    \filldraw (1.75,0) circle (1.5pt) node[below]{};
    
    \node at (0,0.2) {Median};
    \node at (-1.5,0.2) {};
    \node at (-0.75,0.2) {};

\end{tikzpicture}
\end{center}
\label{IllustrationOfContradictionProp1}
\caption{illustrates the contradiction in Proposition 1. If the proposer p is not the median voter, agent p+1 has an incentive to make a proposal, as this would win the vote and make them better off, leading to a contradiction.}
\end{figure}
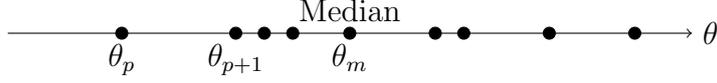

\subsection{Proof of Proposition \ref{proposition:no-two-proposers-N-odd}}

\subsubsection{Lemmata}

We start with two lemmas.

\begin{lemma}\label{lem:interval}
     Assume that there exists a unique median voter $mv$ whose type is $\theta_{mv}$. Then, for every equilibrium, there exists a corresponding equilibrium in which all proposals lie between the proposer's and the median voter's type. In other words, for all $p \in \mathcal{CP}$, $x_p \in [\theta_p, \theta_{mv}]$ if $\theta_p \leq \theta_{mv}$ and $x_p \in [\theta_{mv}, \theta_p]$ otherwise.
\end{lemma}

\begin{proof}
     We can assume without loss of generality that $\theta_p < \theta_{mv}$; otherwise, the same arguments apply due to symmetry. Suppose that there is a proposer $p$ such that $x_p \notin [\theta_p, \theta_{mv}]$. We consider two cases.
    \begin{enumerate}[label=(\roman*)]
        \item If $\theta_{mv} < x_p$, switching to the proposal $\tilde{x}_p = 2\theta_{mv} - x_p$ implies
\begin{align*}
    \mathbb{E}(u_{\theta_p, x_p}) & = -\frac{1}{{(|\mathcal{CP}|-1)^2}} \left( \sum_{i \in \mathcal{CP} \setminus \{p\}} \left( (\theta_{p} - x_{w(i,p)})^2 + \sum_{j \in \mathcal{CP} \setminus \{p, i\}} (\theta_{p} - x_{w(i,j)})^2 \right) \right) \\
    & \leq -\frac{1}{{(|\mathcal{CP}|-1)^2}} \left( \sum_{i \in \mathcal{CP} \setminus \{p\}} \left( (\theta_{p} - \tilde{x}_{w(i,p)})^2 + \sum_{j \in \mathcal{CP} \setminus \{p, i\}} (\theta_{p} - x_{w(i,j)})^2 \right) \right) = \mathbb{E}(u_{\theta_p, 2\theta_{mv}-x_p}),
\end{align*}
where $w(i,j) = \argmin_{k \in \{i,j\}} |x_k - \theta_{mv}|$ is the index of the winning proposal if $x_i$ and $x_j$ are put to vote (where, by convention, we pick the proposer with smaller index if the minimum argument is not unique). There is a slight abuse of notation in the second equation, where $\tilde{x}_{w(i,p)}$ indicates the winning proposal if $\tilde{x}_p$ and $x_i$ are put to vote. The inequality holds because $\tilde{x}_p$ wins against the same proposals of $x_p$, since $|x_p - \theta_{mv}| = |\tilde{x}_p - \theta_{mv}|$. Whenever $x_p$ wins against at least one other proposal (i.e., it is not the most extreme proposal), the inequality is strict. Indeed, as a consequence of $\theta_p < \theta_{mv}$ and $\theta_{mv} < x_p$, we always have $|\theta_p - \tilde{x}_p| < x_p - \theta_p$. By ``most extreme proposal,'' we refer to the proposal that is farthest from the median voter type $\theta_{mv}$. This yields a contradiction, as proposer $p$ has an incentive to change the proposal, and hence this is not an equilibrium. Otherwise, if $x_p$ is the most extreme proposal, the expected utility remains the same for all agents, resulting in a new equilibrium where $x_p < \theta_{mv}$. However, this does not guarantee that $x_p \geq \theta_p$, which is the goal of the next part of the proof.

             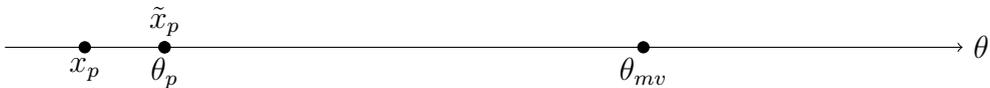
\begin{figure}[htbp]
    \centering
    \begin{tikzpicture}[scale=2.1, decorate, decoration={brace}]
        \draw[->] (-3,0) -- (3,0) node[right]{$\theta$};

        \filldraw (1,0) circle (1pt) node[below]{$\theta_{mv}$};
        \filldraw (-2,0) circle (1pt) node[below]{$\theta_p$};
        \filldraw (2,0) circle (1pt) node[below]{$x_p$};
        \filldraw (0,0) circle (1pt) node[below]{$\tilde{x}_p$};

        \draw[decorate] (1,0.2) -- node[above=4pt] {$|x_p-\theta_{mv}|$} (2,0.2);
        \draw[decorate,decoration={brace, mirror}] (1,0.2) -- node[above=4pt] {$|\tilde{x}_p-\theta_{mv}|$} (0,0.2);
    \end{tikzpicture}
    \caption{Proposer $p$ has an incentive to move his proposal from $x_p$ to $\tilde{x}_p$.}
    \label{fig:number_line}
\end{figure}
        \item
If \( x_p < \theta_p \), switching to the proposal \( \tilde{x}_p = \theta_p \) leads to
\begin{align*}
    \mathbb{E}(u_{\theta_p, x_p}) & = -\frac{1}{{(|\mathcal{CP}|-1)^2}} \left( \sum_{i \in \mathcal{CP} \setminus \{p\}} \left( (\theta_{p} - x_{w(i,p)})^2 + \sum_{j \in \mathcal{CP} \setminus \{p, i\}} (\theta_{p} - x_{w(i,j)})^2 \right) \right)\\
    & \leq -\frac{1}{{(|\mathcal{CP}|-1)^2}} \left( \sum_{i \in \mathcal{CP} \setminus \{p\}} \left( (\theta_{p} - \tilde{x}_{w(i,p)})^2 + \sum_{j \in \mathcal{CP} \setminus \{p, i\}} (\theta_{p} - x_{w(i,j)})^2 \right) \right) = \mathbb{E}(u_{\theta_p, \theta_p}).
\end{align*}
The inequality holds since \( x_p \) can only win against proposals \( x \notin [\theta_p, 2\theta_{mv} - \theta_p] \), which are further away from \( \theta_{mv} \) than \( \theta_p \) itself. Therefore, the new proposal \( \tilde{x}_p \) leads to more utility for the proposer \( p \). For all other pairs of proposals, the outcome does not change. The inequality is strict if there is a proposal \( x_i \), \( i \neq p \), such that \( |x_i - \theta_{mv}| > | \theta_p - \theta_{mv}| \). In this case, this leads to a contradiction. Otherwise, we have found a new equilibrium where \( x_p \in [\theta_p, \theta_{mv}] \).

\begin{figure}[htbp]
\centering
\begin{tikzpicture}[scale=2.1, decorate, decoration={brace}]
    \draw[->] (-3,0) -- (3,0) node[right]{$\theta$};

    \filldraw (1,0) circle (1pt) node[below]{$\theta_{mv}$};
    \filldraw (-2,0) circle (1pt) node[below]{$\theta_p$};
    \filldraw (-2.5,0) circle (1pt) node[below]{$x_p$};
    \filldraw (-2,0) circle (1pt) node[above]{$\tilde{x}_p$};

    
\end{tikzpicture}
\caption{Proposer \( p \) has an incentive to move his proposal from \( x_p \) to \( \tilde{x}_p \).}
\label{fig:number_line}
\end{figure}

We remark that the proofs of both points are complementary. Indeed, in the first point, we proved that if \( \theta_{mv} < x_p \), then either we reach a contradiction or we can find a new equilibrium where \( x_p < \theta_{mv} \). This alone does not guarantee that \( \theta_{p} \leq x_p \). If this is not true, then the second point shows how to find another equilibrium for which this holds. We emphasize this as it is important to provide a complete proof.
\end{enumerate}

\end{proof}

\begin{lemma}\label{lem:peak}
    Assume again that there exists a unique median voter $mv$ whose type is $\theta_{mv}$. Furthermore, the equilibrium is such that all proposals satisfy $x_p \in [\min\{\theta_p, \theta_{mv}\}, \max\{\theta_p, \theta_{mv}\}]$ and no two proposals coincide, that is, $x_p\neq x_q$ for all $p, q \in \mathcal{CP}$. Then $x_p=\theta_p$ for all proposals.
\end{lemma}

\begin{proof}
    We denote by $p_1, \dots, p_s$ the $s$ agents making a proposal. We start by noticing that there cannot exist $p_i \neq p_j$ such that $|\theta_{mv} - x_{p_i}|=|\theta_{mv} - x_{p_j}|$ and $x_{p_i} \neq x_{p_j}$. If this were the case, then both the proposer $p_i$ and $p_j$ could profit by deviating from their current strategy, moving their proposal by a small amount $\epsilon >0$ towards the median. This ensures that their proposal wins if the two proposals are put to vote against each other.
    
    Therefore, there exists a unique proposal $x_{p_c}=\argmin_{x_{p} \in \{x_{p_1},\dots, x_{p_s}\} } |x_p - \theta_{mv}|$ that is closest to the median. We assume by contradiction that $x_{p_c} \neq \theta_{p_c}$. Introduce the second closest proposal $x_{p_{2c}}= \argmin_{x_{p} \in \{x_{p_1},\dots, x_{p_s}\} \setminus x_{p_c} } |x_p - \theta_{mv}|$. Since $|x_{p_c} - \theta_{mv}| < | x_{p_{2c}} - \theta_{mv}|$ by assumption, the proposer $p_c$ profits by moving its proposal $x_{p_c}$ by a small amount $\epsilon < | x_{p_{2c}} - \theta_{mv}| - | x_{p_{c}} - \theta_{mv}|$ towards its peak $\theta_{p_c}$. This does not change the outcome of the vote, but increases the utility of $p_c$ since $\widetilde{x}_{p_c}=x_{p_c}+\epsilon \cdot \text{sgn}( \theta_{p_c} - x_{p_c})$ wins against all other proposals. This is a contradiction, hence $x_{p_c}=\theta_{p_c}$.

    In the same way, we can show that $x_p=\theta_p$ for any $p \in \mathcal{CP}$.
\end{proof}

\subsubsection{Proof of Proposition \ref{proposition:no-two-proposers-N-odd}}

Suppose by contradiction that there exists an equilibrium with exactly two proposers $p_1$ and $p_2$, and assume w.l.o.g. that $p_1 < p_2$. 
    
    If $x_{p_1}=x_{p_2}$, the result of the voting is $x_{p_1}$. Hence, $p_1$ will prefer to vote instead of making a proposal by the tie-breaking rule, as voting will lead to the same utility since the outcome remains $x_{p_1}$. This is a contradiction. Thus, we can assume that $x_{p_1}\neq x_{p_2}$.
    
    Moreover, if the peak of the median voter is proposed, i.e. $x_{p_1}=\theta_{mv}$ or $x_{p_2}=\theta_{mv}$, the other proposer cannot make any proposal that beats the median voter's proposal. Hence, choosing not to propose and to vote instead yields the same welfare. By the tie-breaking rule, voting is preferred, and therefore, there is no equilibrium where the median $\theta_{mv}$ is proposed. We distinguish three cases:

    \begin{enumerate}[label=(\roman*)]
        \item 
            Assume that $\theta_{p_1}< \theta_{p_2} < \theta_{mv}$. If $x_{p_1} < \theta_{p_1}$, moving the proposal to $\theta_{p_1}$ brings it closer to the median $\theta_{mv}$ and therefore increases its chance of winning against $x_{p_2}$. Moreover, as the distance to its own type $\theta_{p_1}$ is zero, this can only improve the welfare of $p_1$. Hence, this either yields a contradiction, and we are done, or we obtain a second equilibrium where $x_{p_1}\geq\theta_{p_1}$. The same applies to the case $x_{p_2}< \theta_{p_2}$. Thus, we can assume that
            \begin{equation*}
                \theta_{p_1} \leq x_{p_1} \text{ and } \theta_{p_2} \leq x_{p_2}.
            \end{equation*}
            Moreover, from \Cref{lem:interval} we know that $x_{p_1} \in [ \theta_{p_1}, \theta_{mv}]$ and $x_{p_2} \in [ \theta_{p_2}, \theta_{mv}]$.  If $\theta_{p_2} \leq x_{p_2} < x_{p_1} < \theta_{mv}$, the proposal $x_{p_1}$ wins the vote and is implemented. However, if $p_1$ decides to vote instead, $x_{p_2}$ is implemented, which is closer to $\theta_{p_1}$ and therefore yields more utility for $p_1$, a contradiction. If instead $|x_{p_2} - \theta_{mv}| < | x_{p_1} - \theta_{mv}|$, $x_{p_2}$ will win and be implemented. For $p_1$ to win the vote with $x_{p_1}$, $p_1$ has to move the policy at least to $x_{p_2}$. However, this leads to the same utility as not proposing or less. Hence, the best strategy for $p_1$ is to vote, which yields the same utility as making a proposal. There is no equilibrium with such an ordering of types.
        \item 
            Assume that $\theta_{p_1}< \theta_{mv} < \theta_{p_2}$. We consider the proposer whose proposal is further away from the median voter $\theta_{mv}$. Assume without loss of generality that $|x_{p_1} - \theta_{mv}| > | x_{p_2} - \theta_{mv}|$ and thus $x_{p_2}$ is implemented as it wins the vote. It is important to note that at equilibrium, we must have $\theta_{mv} < x_{p_2}$. If this were not the case, then agent $p_2$ could increase their utility by deviating to the symmetric point $2\theta_{mv} - x_{p_2}$, which preserves the probability of winning the vote but yields a strictly higher utility. Now, if $p_1$ were to propose $\tilde{x}_{p_1}=2\theta_{mv} - x_{p_2} + \epsilon$ for some $0 < \epsilon < | x_{p_2} - \theta_{mv} |$ (recall that $x_{p_2} \neq \theta_{mv}$), then
            \begin{equation*}
                \mathbb{E}(u_{\theta_{p_1}, \tilde{x}_{p_1}})= -(\tilde{x}_{p_1} - \theta_{p_1})^2 > -(x_{p_2} - \theta_{p_1})^2 =\mathbb{E}(u_{\theta_{p_1}, x_{p_1}}).
            \end{equation*}
            Hence, $p_1$ can improve its welfare by changing the proposal, and therefore, this cannot be an equilibrium.
            
            \begin{minipage}{\linewidth}
                    \vspace{10pt}
                    \centering
                    \begin{tikzpicture}
                        \draw (0,0) -- (10,0);
                        \draw (1,-0.2) -- (1,0.2); \node at (1, -0.5) {$\theta_{p_1}$};
                        \draw (2,-0.2) -- (2,0.2); \node at (2, -0.5) {$x_{p_1}$};
                         \draw (3.5,-0.2) -- (3.5,0.2); \node at (3.5, -0.5) {$\tilde{x}_{p_1}$};
                        \draw [->] (2, 0.1) to [bend left] (3.5, 0.1);
                        \draw (5,-0.2) -- (5,0.2); \node at (5, -0.5) {$\theta_{mv}$};
                        \draw (7,-0.2) -- (7,0.2); \node at (7, -0.5) {$x_{p_2}$};
                        \draw (9,-0.2) -- (9,0.2); \node at (9, -0.5) {$\theta_{p_2}$};
                    \end{tikzpicture}
                    
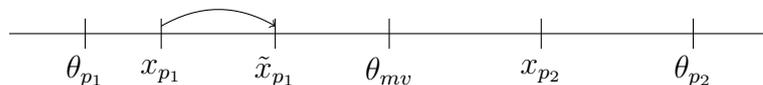
\captionof{figure}{Welfare can be improved by moving one's own proposal towards the median, ensuring it is implemented.}
                    \label{fig:n=2_deviation}
                    \vspace{15pt}
            \end{minipage}
        \item 
            Assume that $\theta_{mv} < \theta_{p_1} < \theta_{p_2}$. This situation is symmetric to case $(i)$ and thus the same arguments apply.
    \end{enumerate}

\subsection{Proof of Theorem \ref{theorem:Uniqueness_of_N_odd}}
The goal of this section is to demonstrate that in the PoV procedure, when the number of agents $N$ is odd, the only equilibrium is the one characterized in Theorem \ref{theorem:vop-theorem}.

\subsubsection{No equilibrium with $N$ proposals}
We begin with the following intermediate result:

\begin{proposition}\label{Proposition:No_eq_everyone_proposing}
    Given an odd number of agents $N$, there is no equilibrium in which every agent proposes.
\end{proposition}

\begin{proof}
Suppose, by contradiction, that there exists an equilibrium in which every agent is proposing.

Assume that in this equilibrium, each agent $i \in [N]$ of type $\theta_i$ proposes policy $x_i$. If there exists an agent $j$ such that $x_j \neq \theta_j$, then agent $j$ could increase his/her utility by instead proposing $x_j = \theta_j$. Given that all other agents are also proposing, this deviation would strictly improve his/her expected utility. Therefore, it must be that in this equilibrium, $x_i = \theta_i$ for all $i \in [N]$.

Under this strategy profile, the utility of agent $k$ is:
\[
\mathbb{E}(u_{\theta_{k}, \theta_k}) = -\frac{1}{N} \sum_{i=1}^N (\theta_k - \theta_i)^2.
\]

Since this is presumed to be an equilibrium, no agent should have a profitable deviation. We have already shown that no deviation in the proposals improves utility. We now consider the possibility of deviating by choosing to vote instead.

Suppose agent $k$ deviates by voting. His or her expected utility in that case would be:
\[
\mathbb{E}(u_{\theta_{k}, v}) = -\frac{1}{(N-1)(N-2)} \sum_{i\neq k} \sum_{j\notin\{i,k\}} (\theta_k - \theta_{w\{i,j\}})^2,
\]
where $w\{i,j\} = \argmin_{s\in\{i,j\}}|\theta_k - \theta_s|$ (choosing the smaller index again if it is not unique).

Because all other agents are proposing and the deviating agent gets to vote between two uniformly chosen proposals, he/she will always choose the one closer to their own type.

For voting not to be a profitable deviation, it must hold for every agent $k$ that:
\[
-\frac{1}{N} \sum_{i=1}^N (\theta_k - \theta_i)^2
>
-\frac{1}{(N-1)(N-2)} \sum_{i\neq k} \sum_{j\notin\{i,k\}} (\theta_k - \theta_{w\{i,j\}})^2.
\]

Assume without loss of generality that $|\theta_1 - \theta_m| \leq |\theta_N - \theta_m|$. That is, between the two extreme agents, agent $1$ is at most as far from the median as agent $N$. We will now prove that agent $1$ has a strict incentive to deviate by voting.

Specifically, we show that:
\[
\begin{aligned}
    \mathbb{E}(u_{\theta_{1}, \theta_1}) = -\frac{1}{N} \sum_{i=2}^N (\theta_1 - \theta_i)^2 &< -\frac{1}{\binom{N-1}{2}} \sum_{i = 2}^{N-1} \sum_{j=i+1}^{N} (\theta_1 - \theta_i)^2 = \mathbb{E}(u_{\theta_{1}, v}) \\
    \iff -\frac{1}{N} \sum_{i=2}^N (\theta_1 - \theta_i)^2 &< -\frac{1}{\binom{N-1}{2}} \sum_{i = 2}^{N-1} (N - i) (\theta_1 - \theta_i)^2.
\end{aligned}
\]

The second equality holds since agent $1$ will always choose the proposal closest to their type.

We now simplify the inequality to a more tractable form:
\[
\begin{aligned}
-\frac{1}{N} \sum_{i=2}^N (\theta_1 - \theta_i)^2 
&< -\frac{1}{\binom{N-1}{2}} \sum_{i=2}^{N-1} (N - i) (\theta_1 - \theta_i)^2 \\
\iff \frac{1}{N} \sum_{i=2}^N (\theta_1 - \theta_i)^2 
&> \frac{1}{\binom{N-1}{2}} \sum_{i=2}^{N-1} (N - i) (\theta_1 - \theta_i)^2 \\
\iff (N-1)(N-2) \sum_{i=2}^N (\theta_1 - \theta_i)^2 
&> 2N \sum_{i=2}^{N-1} (N - i) (\theta_1 - \theta_i)^2 \\
\iff (N^2 - 3N + 2) \sum_{i=2}^N (\theta_1 - \theta_i)^2 
&> 2N \sum_{i=2}^{N-1} (N - i) (\theta_1 - \theta_i)^2.
\end{aligned}
\]

Expanding the right-hand side and regrouping:
\[
\begin{aligned}
(N^2 - 3N + 2)(\theta_1 - \theta_N)^2 &+ \sum_{i=2}^{N-1} \left[ (N^2 - 3N + 2) - 2N(N - i) \right](\theta_1 - \theta_i)^2 \\
&= (N^2 - 3N + 2)(\theta_1 - \theta_N)^2 + \sum_{i=2}^{N-1} c_i(\theta_1 - \theta_i)^2 > 0,
\end{aligned}
\]
where \( c_i := 2N\left(i - \frac{N+3}{2}\right) + 2 \).

We now prove that the inequality holds. To do so, we first analyze the coefficients  
$c_i$. 
We begin by determining the value of \( i \) for which \( c_i = 0 \). Solving,
\[
c_i = 0 \iff 2N\left(i - \frac{N+3}{2}\right) + 2 = 0 \iff i = \frac{N+3}{2} - \frac{1}{N}.
\]
From this, we can deduce the sign of \( c_i \) based on the value of \( i \):
\[
\begin{cases}
c_i < 0 & \text{if } i < \frac{N+3}{2} - \frac{1}{N}, \\
c_i > 0 & \text{if } i > \frac{N+3}{2} - \frac{1}{N}.
\end{cases}
\]

Since \( N \) is assumed to be odd, the expression \( \frac{N+3}{2} - \frac{1}{N} \) lies strictly between two consecutive integers. Specifically, observe that
\[
\frac{N+3}{2} - 1 = \frac{N+1}{2} = m.
\]
Therefore, the sign conditions for \( c_i \) can be equivalently written as:
\[
\begin{cases}
c_i < 0 & \text{if } i \leq m, \\
c_i > 0 & \text{if } i \geq m+1.
\end{cases}
\]

With this information about the sign of the coefficients \( c_i \), we can now rewrite the expression as:
\begin{equation}\label{result}
(N^2 - 3N + 2)(\theta_1 - \theta_N)^2 
+ \sum_{i=2}^{m} c_i(\theta_1 - \theta_i)^2 
+ \sum_{i=m+1}^{N-1} c_i(\theta_1 - \theta_i)^2 > 0.
\end{equation}
At this point, we make the following observations:
\begin{itemize}
    \item Since \( |\theta_1 - \theta_m| \leq |\theta_N - \theta_m| \), it follows that
    \[
    (\theta_1 - \theta_N)^2 \geq 4(\theta_1 - \theta_m)^2,
    \]
    and therefore,
    \[
    (N^2 - 3N + 2)(\theta_1 - \theta_N)^2 \geq 4(N^2 - 3N + 2)(\theta_1 - \theta_m)^2.
    \]
    
    \item For all \( i \in \{2, \dots, m\} \), we have \( |\theta_1 - \theta_i| \leq |\theta_1 - \theta_m| \), and since \( c_i < 0 \), this implies
    \[
    \sum_{i=2}^{m} c_i (\theta_1 - \theta_i)^2 > \sum_{i=2}^{m} c_i (\theta_1 - \theta_m)^2.
    \]

    \item For all \( i \in \{m+1, \dots, N-1\} \), we have \( |\theta_1 - \theta_i| > |\theta_1 - \theta_m| \), and since \( c_i > 0 \), this implies
    \[
    \sum_{i=m+1}^{N-1} c_i (\theta_1 - \theta_i)^2 > \sum_{i=m+1}^{N-1} c_i (\theta_1 - \theta_m)^2.
    \]
\end{itemize}

Thanks to these observations, if we can show that
\begin{equation}\label{equation:implyresult}
    4(N^2 - 3N + 2) + \sum_{i=2}^{m} c_i + \sum_{i=m+1}^{N-1} c_i > 0,
\end{equation}
then this will imply inequality~\eqref{result}.

We now proceed to prove inequality~\eqref{equation:implyresult}, which can be rewritten as:
\[
\sum_{i=2}^{m} \left[2N\left(i - \frac{N+3}{2}\right) + 2\right] 
+ \sum_{i=m+1}^{N-1} \left[2N\left(i - \frac{N+3}{2}\right) + 2\right] 
+ 4(N - 1)(N - 2) > 0.
\]

Let \( N \in \mathbb{N} \) be an odd integer with \( N \geq 3 \), and remember \( m := \frac{N+1}{2} \). Observe that the two sums together cover the full index range from \( i=2 \) to \( i=N-1 \). Hence, the expression simplifies to:
\[
\sum_{i=2}^{N-1} \left[ 2N\left(i - \frac{N+3}{2}\right) + 2 \right] + 4(N-1)(N-2).
\]

Let us compute the sum:
\begin{align*}
\sum_{i=2}^{N-1} \left[ 2N\left(i - \frac{N+3}{2}\right) + 2 \right]
&= \sum_{i=2}^{N-1} \left( 2Ni - N(N+3) + 2 \right) \\
&= \sum_{i=2}^{N-1} 2Ni - (N-2)N(N+3) + 2(N - 2),
\end{align*}
since there are \( N - 2 \) terms in the sum.

The partial sum \( \sum_{i=2}^{N-1} i \) is:
\[
\sum_{i=2}^{N-1} i = \sum_{i=1}^{N-1} i - 1 = \frac{(N-1)N}{2} - 1.
\]
Therefore,
\begin{align*}
\sum_{i=2}^{N-1} 2Ni &= 2N \left( \frac{(N-1)N}{2} - 1 \right) = N^2(N-1) - 2N, \\
\text{and thus, } \sum_{i=2}^{N-1} \left[ 2N\left(i - \frac{N+3}{2}\right) + 2 \right]
&= N^2(N-1) - 2N - (N-2)N(N+3) + 2(N - 2). \\
&=  (N^3 - N^2) - 2N - (N^3 + N^2 - 6N) + 2N - 4 \\
&= -2N^2 + 6N - 4.
\end{align*}
Now add the remaining term:
\[
4(N-1)(N-2) = 4(N^2 - 3N + 2) = 4N^2 - 12N + 8.
\]
Hence, the total expression becomes:
\begin{align*}
-2N^2 + 6N - 4 + 4N^2 - 12N + 8 &= 2N^2 - 6N + 4.
\end{align*}
We now show that:
\[
2N^2 - 6N + 4 > 0 \quad \text{for odd } N \geq 3.
\]
Indeed,
\[
N^2 - 3N + 2 = (N-2)(N-1) > 0 \quad \text{for odd } N \geq 3.
\]
Thus, the expression is strictly positive for all odd \( N \geq 3 \), and the inequality is proven.

We can then conclude that, since agent $1$ has an incentive to deviate to voting, this cannot be an equilibrium. This proves the Proposition. 
\end{proof}

It is important to notice that the assumption that \( N \) is odd is necessary in order to guarantee that the threshold at which the coefficient \( c_i \) changes sign occurs exactly at the median index \( m \). On the other hand, it is also crucial to choose the most extreme agent who is \emph{closest} to the median agent. Indeed, observe that one can construct a counterexample where the most extreme agent does not have an incentive to deviate to voting, if it is not the extreme agent closest to the median. 

We propose a simple example of such a scenario here.

Consider three agents with types \( 0 \), \( 9 \), and \( 10 \). We know that in an equilibrium in which everyone is proposing, each agent has to propose their own type. Then, the agent of type \( 0 \) has no incentive to deviate to voting. His expected utility from proposing is given by:
\[
-\frac{1}{3}(81 + 100)\approx -60,
\]
since the two squared distances from the median (which is 9) are \( 81 = (0 - 9)^2 \) and \( 100 = (0 - 10)^2 \). On the other hand, if he deviates and votes (i.e., votes 9), he would obtain utility:
\[
-(0 - 9)^2 = -81
\]
with certainty. This deviation is not beneficial, as:
\[
-81 < -60.
\]
However, this profile is not an equilibrium either, since the agent of type \( 10 \) does have an incentive to deviate to voting: his utility would improve from $-\frac{1}{3}(81 + 100) = -\frac{181}{3} $ to $-(10 - 9)^2 = -1.$

Observe that this result, combined with the findings of Theorem \ref{theorem:vop-theorem} and Proposition \ref{proposition:no-two-proposers-N-odd}, ensures that when \( N = 3 \), the equilibrium described in Theorem \ref{theorem:vop-theorem} is unique. Specifically, we have demonstrated that no equilibrium exists with three proposers. Additionally, Proposition \ref{proposition:no-two-proposers-N-odd} confirms the absence of an equilibrium with two proposers. Finally, with only one proposer, the sole feasible equilibrium is the one characterized by Theorem \ref{theorem:vop-theorem}. This confirms uniqueness in this particular scenario.

\subsubsection{Proofs for \(N\)-\(k\) Proposals}

We now proceed with the formal proof of the uniqueness of the equilibrium as stated in Theorem \ref{theorem:vop-theorem}. After establishing in Proposition \ref{Proposition:No_eq_everyone_proposing} that no equilibrium exists where everyone is proposing, we extend this result to show that there is also no equilibrium with \(N-1\) agents proposing. Specifically, we prove the following proposition:

\begin{proposition}\label{proposition:no_eq_N-1_Proposing}
Given an odd number of agents \(N\), there is no equilibrium in which \(N-1\) agents are proposing.
\end{proposition}

\begin{proof}
To prove this statement, it is necessary to show that for every strategy profile in which \(N-1\) agents are proposing, there exists at least one agent who has a profitable deviation from their current strategy.

We proceed by contradiction. Assume that there exists an equilibrium where \(N-1\) agents are proposing. We will show that this leads to a contradiction.

Let \(\mathbf{v}\) be the index of the agent who is voting, i.e., the only agent who is not proposing. All other agents \(i \in [N] \setminus \{\mathbf{v}\}\) are proposing a policy \(x_i\).

Notice that at this equilibrium, the utility realized by the voter \(\mathbf{v}\) is given by:

\[
\mathbb{E}(u_{\theta_\mathbf{v}, v}) = -\frac{1}{(N-1)(N-2)} \sum_{i \neq \mathbf{v}} \sum_{j \notin \{i, \mathbf{v}\}} (\theta_\mathbf{v} - x_{w\{i,j\}})^2,
\]

where \(w\{i,j\} = \argmin_{s \in \{i,j\}} |\theta_\mathbf{v} - x_s|\), where we choose the smaller index if necessary. Essentially, in this scenario, two proposals are randomly selected from all the proposals made by the other \(N-1\) agents, and then the voting agent \(\mathbf{v}\) chooses the one that is closest to his type.

On the other hand, the utility of an agent \(k \in [N] \setminus \{\mathbf{v}\}\) is given by:

\[
\mathbb{E}(u_{\theta_k, x_k}) = -\frac{1}{(N-1)(N-2)} \sum_{i \neq \mathbf{v}} \sum_{j \notin \{i, \mathbf{v}\}} (\theta_k - x_{w\{i,j\}})^2,
\]

where \(w\{i,j\} = \argmin_{s \in \{i,j\}} |\theta_\mathbf{v} - x_s|\). In this case, agent \(k\) proposes \(x_k\), and then two proposals are selected uniformly at random from all the proposals, and these two are put to a vote. Since \(\mathbf{v}\) is the only voter, he chooses the policy that is closest to his type.
Consider now the multiset of offers
\[
\{x_i : i \in [N] \setminus \{\mathbf{v}\}\}.
\]

Note that we are using the concept of a multiset here (for more details on this concept, see the proof of Proposition \ref{proposition:no_eq_odd_voting}). This will allow us to store single offers with a multiplicity greater than 1. If there exists just one offer that is the farthest from \(\theta_{\mathbf{v}}\), then, regardless of which proposal it is compared against, this proposal will always lose. Specifically, the agent making such a proposal has at least the same utility from voting, as his proposal is never chosen. In the worst case, his vote does not influence the decision. Consequently, he has an incentive to vote, given that if an agent is indifferent between voting and proposing, he will choose to vote.

This proves that there is no single proposal that is the worst. Hence, we know that there are at least two worst proposals. By definition, these must be equidistant from the type of the voter. Furthermore, at equilibrium, there cannot be just two identical worst proposals. Indeed, if there are just two identical worst proposals, they will always lose against every other proposal except when compared against each other. If these proposals are identical, then the outcome remains unchanged, and switching to voting would not decrease their utility. This proves that, at equilibrium, there are at least two worst proposals equidistant from the type of the voter and symmetric with respect to his/her type.

\begin{center}
\begin{tikzpicture}

\draw[->] (0,0) -- (14,0);

\filldraw (2,0) circle (3pt) node[above]{$x_{worst_1}$};
\filldraw (7,0) circle (3pt) node[above]{$\theta_{\mathbf{v}}$};
\filldraw (12,0) circle (3pt) node[above]{$x_{worst_2}$};

\draw[decorate, decoration={brace, mirror, amplitude=10pt}, yshift=-5pt]
(2,0) -- node[below=12pt] {All other proposals \( x_i \) lie within this interval.
} (12,0);

\end{tikzpicture}
\end{center}

Let us denote the agents making two of the worst offers as \( \text{worst}_1 \) and \( \text{worst}_2 \). Without loss of generality, we assume that \( x_{\text{worst}_1} < x_{\text{worst}_2} \), as depicted in the figure.

For this scenario to represent an equilibrium, the following conditions must hold:
\begin{enumerate}
    \item \( \theta_{\text{worst}_1} \leq x_{\text{worst}_1} \leq \theta_{m} \);
    \item \( \theta_{\text{worst}_2} \geq x_{\text{worst}_2} \geq \theta_{m} \).
\end{enumerate}

These conditions are a consequence of Lemma \ref{lem:interval}. For completeness, we provide a short proof here. They can be proven in an identical manner, so we will prove only one of them, specifically \( \theta_{\text{worst}_1} \leq x_{\text{worst}_1} \).

Assume, for contradiction, that this is not the case. If \( \theta_{\text{worst}_1} > x_{\text{worst}_1} \), then agent \( \text{worst}_1 \) could improve his expected utility by moving his proposal to \( \theta_{\text{worst}_1} \). This adjustment would increase his chances of winning and potentially prevent his proposal from being the worst, as it would be closer to the voter's type. Additionally, winning would yield a higher utility since the proposal aligns more closely with his preferences. The existence of this profitable deviation contradicts the assumption of equilibrium.

Notice that these two facts guarantee that there is no equilibrium where the voter is either agent \( 1 \) or agent \( N \). In both cases, the voter would lack either a predecessor or a successor. Hence, this proves that there is no equilibrium with the first or last agent being the only voter.

After these observations we are at the point where the situation is like the one depicted in the following picture.
\begin{center}
\begin{tikzpicture}

\draw[->] (0,0) -- (14,0);

\filldraw (0.5,0) circle (3pt) node[above]{$\theta_{worst_1}$};
\filldraw (2,0) circle (3pt) node[above]{$x_{worst_1}$};
\filldraw (7,0) circle (3pt) node[above]{$\theta_{\mathbf{v}}$};
\filldraw (12,0) circle (3pt) node[above]{$x_{worst_2}$};
\filldraw (13.5,0) circle (3pt) node[above]{$\theta_{worst_2}$};

\draw[decorate, decoration={brace, mirror, amplitude=10pt}, yshift=-5pt]
(2,0) -- node[below=12pt] {All other proposals \( x_i \) lie within this interval.
} (12,0);

\end{tikzpicture}
\end{center}

We now prove that this scenario cannot represent an equilibrium. Specifically, agent \(\text{worst}_1\) has an incentive to propose \(x_{\text{worst}_1} + \varepsilon\) for a sufficiently small \(\varepsilon > 0\). This adjustment makes his proposal better than \(x_{\text{worst}_2}\) but still worse than the other proposals against which it was already losing, given that $\varepsilon$ is sufficiently small.

Now, whenever this new policy is put to a vote against the other policies, different from the worst ones, it will still lose. Therefore, the part of the expected utility for these combinations does not change for agent \(\text{worst}_1\). However, when his proposal is voted against that of \(\text{worst}_2\), it now wins the vote and this policy is more favorable than the random choice previously made between \(x_{\text{worst}_1}\) and \(x_{\text{worst}_2}\). This can be seen as paying a small price \(\varepsilon\) to avoid the possibility that \(x_{\text{worst}_2}\) could be implemented. This results in a profitable deviation for agent \(\text{worst}_1\), reaching a contradiction. 

Thus, we have shown that in any configuration where \(N-1\) agents are proposing, there exists at least one agent with an incentive to deviate from the proposed strategy. This deviation undermines the stability of the configuration, leading to a contradiction with the assumption of equilibrium. Therefore, no equilibrium exists under the given conditions, completing the proof.

\end{proof}

We now proceed with the next step toward establishing the uniqueness of the equilibrium stated in Theorem~\ref{theorem:vop-theorem}. Specifically, we prove the following:

\begin{proposition}\label{proposition:no_eq_odd_voting}
Let \(N\) and \(k\) be odd integers, $k < N$. Then, there exists no equilibrium in which exactly \(N-k\) agents are proposing.
\end{proposition}

\begin{proof}
The case \(k = 1\) was established in the previous proposition. The argument here extends the same core reasoning.

Assume, for contradiction, that there exists an equilibrium in which \(N-k\) agents are proposing. Then, \(k\) agents are voting. Denote the set of voting agents by \(\mathcal{CV}\), and let their types be ordered as
\[
\theta_{\mathbf{v}_1} < \theta_{\mathbf{v}_2} < \dots < \theta_{\mathbf{v}_k}.
\]
Since \(k\) is odd, the median voter is uniquely defined and his type is denoted as \(\theta_{\mathbf{v}_m}\).

\begin{center}
\begin{tikzpicture}
\draw[->] (0,0) -- (14,0);
\filldraw (4,0) circle (3pt) node[above]{$\theta_{\mathbf{v}_1}$};
\filldraw (4.5,0) circle (3pt) node[above]{$\theta_{\mathbf{v}_2}$};
\filldraw (6,0) circle (0pt) node[above]{$\dots$};

\filldraw (7,0) circle (3pt) node[above]{$\theta_{\mathbf{v}_m}$};
\filldraw (8,0) circle (0pt) node[above]{$\dots$};

\filldraw (10,0) circle (3pt) node[above]{$\theta_{\mathbf{v}_{k-1}}$};
\filldraw (13,0) circle (3pt) node[above]{$\theta_{\mathbf{v}_k}$};
\end{tikzpicture}
\end{center}

Let \(\mathcal{P} = \{x_i \mid i \in [N] \setminus \mathcal{CV} \}\) be the \emph{multiset} of proposals. Note that we are using a multiset here, which is an extension of a set which can host single elements with a multiplicity higher than 1. This is necessary as two distinct agents might submit the same proposal and we want to be able to account for that.

Suppose there exists a single proposal that is the furthest from \(\theta_{\mathbf{v}_m}\). Then it always loses against any other kind of proposal, regardless of which other proposal it is put against. The agent who submitted it would strictly prefer to vote, since voting does not decrease the expected utility (and is chosen when indifferent). Thus, no such unique worst proposal can exist in equilibrium.

Therefore, there must be at least two worst proposals. Among these, at least two must be symmetric with respect to \(\theta_{\mathbf{v}_m}\). Indeed, if all the worst proposals coincide at the same value, then the agents making them have no influence on the outcome: their proposals always lose unless compared against each other, and the outcome does not depend on which identical proposal is selected. In such a case, deviating to voting would not decrease their expected utility, contradicting the assumption of equilibrium.

Hence, there must exist at least two distinct worst proposals that are symmetric around \(\theta_{\mathbf{v}_m}\). This symmetry is necessary because if the proposals were not equidistant from the median voter's type, the one closer to \(\theta_{\mathbf{v}_m}\) would win more often, and thus could not be equally bad. Therefore, among the worst proposals, we can always identify two that are distinct and symmetric with respect to the median voter's type.

\begin{center}
\begin{tikzpicture}
\draw[->] (0,0) -- (14,0);
\filldraw (2,0) circle (3pt) node[above]{$x_{\text{worst}_1}$};
\filldraw (7,0) circle (3pt) node[above]{$\theta_{\mathbf{v}_m}$};
\filldraw (12,0) circle (3pt) node[above]{$x_{\text{worst}_2}$};
\draw[decorate, decoration={brace, mirror, amplitude=10pt}, yshift=-5pt]
(2,0) -- node[below=12pt] {All proposals \(x_i\) lie in this interval.} (12,0);
\end{tikzpicture}
\end{center}

Let \(\text{worst}_1\) and \(\text{worst}_2\) be the agents submitting \(x_{\text{worst}_1}\) and \(x_{\text{worst}_2}\), respectively, with \(x_{\text{worst}_1} < x_{\text{worst}_2}\).

For this to be an equilibrium, it must be that:
\begin{align}
\theta_{\text{worst}_1} &\leq x_{\text{worst}_1}\leq \theta_{\mathbf{v}_m}, \label{condition_1} \\
\theta_{\text{worst}_2} &\geq x_{\text{worst}_2}\geq \theta_{\mathbf{v}_m}. \label{condition_2}
\end{align}

Conditions \eqref{condition_1} and \eqref{condition_2} are a direct consequence of Lemma \ref{lem:interval}.
Thus, we must have:
\[
\theta_{\text{worst}_1} \leq x_{\text{worst}_1} < \theta_{\mathbf{v}_m} < x_{\text{worst}_2} \leq \theta_{\text{worst}_2}.
\]

The degenerate case \(x_{\text{worst}_1} = \theta_{\mathbf{v}_m} = x_{\text{worst}_2}\) contradicts the earlier finding that the two worst proposals must be distinct.

These inequalities also imply that no equilibrium can exist where the first or last \((k+1)/2\) agents are all voters, since in that case either \eqref{condition_1} or \eqref{condition_2} would be violated.

\begin{center}
\begin{tikzpicture}
\draw[->] (0,0) -- (14,0);
\filldraw (0.5,0) circle (3pt) node[above]{$\theta_{\text{worst}_1}$};
\filldraw (2,0) circle (3pt) node[above]{$x_{\text{worst}_1}$};
\filldraw (7,0) circle (3pt) node[above]{$\theta_{\mathbf{v}_m}$};
\filldraw (12,0) circle (3pt) node[above]{$x_{\text{worst}_2}$};
\filldraw (13.5,0) circle (3pt) node[above]{$\theta_{\text{worst}_2}$};

\end{tikzpicture}
\end{center}

We now show that this scenario cannot represent an equilibrium. Specifically, agent \(\text{worst}_1\) has an incentive to slightly adjust their proposal to \(x_{\text{worst}_1} + \varepsilon\), for a sufficiently small \(\varepsilon > 0\). Suppose that \(\varepsilon\) is chosen such that the new proposal lies between \(x_{\text{worst}_1}\) and \(x_{\text{worst}_2}\), i.e., it is better than \(x_{\text{worst}_2}\) but still worse than all other proposals against which it was already losing.

Under this adjustment, whenever the new proposal is compared against policies other than the set of worst ones, it still loses. Hence, the expected utility from those pairs remains unchanged for agent \(\text{worst}_1\). However, when the adjusted proposal is compared against \(x_{\text{worst}_2}\) and all other worst proposals, it now wins. This new policy yields a strictly higher utility than the random outcome previously resulting from the tie between \(x_{\text{worst}_1}\) and \(x_{\text{worst}_2}\).

Formally, we have:
\[
-(\theta_{\text{worst}_1} - (x_{\text{worst}_1} + \varepsilon))^2 >
-\frac{(\theta_{\text{worst}_1} - x_{\text{worst}_1})^2 + (\theta_{\text{worst}_1} - x_{\text{worst}_2})^2}{2}.
\]
Let \( a = \theta_{\text{worst}_1} - x_{\text{worst}_1} \) and \( b = \theta_{\text{worst}_1} - x_{\text{worst}_2} \). Then the inequality becomes:
\[
-(a - \varepsilon)^2 > -\frac{a^2 + b^2}{2} \quad \Leftrightarrow \quad \varepsilon(\varepsilon - 2a) < \frac{b^2 - a^2}{2}.
\]
Since \(|a| < |b|\), it follows that \(b^2 - a^2 > 0\), and the inequality holds for \(\varepsilon > 0\) and small enough. 

This establishes that agent \(\text{worst}_1\) has a profitable deviation, contradicting the assumption of equilibrium.

Therefore, in any configuration where \(N-k\) agents are proposing, there exists at least one agent with an incentive to deviate. This undermines the stability of the configuration and leads to a contradiction. We conclude that no such equilibrium can exist, completing the proof.

\end{proof}

We now provide the proof for the case in which $k$ is even.
\begin{proposition}\label{proposition:no_eq_even_voting}
Let \(N\) be an odd integer, and let \(k\) be an even integer such that $0 < k < N - 1$. Then, there does not exist an equilibrium in which exactly $N - k$ agents are proposing.
\end{proposition}

We exclude the case $k = 0$, as it has already been proven separately: in that case, we showed that no equilibrium exists with all agents proposing. We also exclude the case $k = N - 1$, as Theorem \ref{theorem:vop-theorem} proves the existence (and unique characterization) of an equilibrium for that configuration.

\begin{proof}
We proceed by contradiction. Suppose there exists an equilibrium in which exactly $N - k$ agents are proposers, with $k$ even and $0 < k < N - 1$.

Let $\mathcal{CP}$ denote the set of proposers and $\mathcal{CV}$ the set of voters, so that $|\mathcal{CV}| = k$ is even. Because $k$ is even, there is no unique median voter. Denote the left and right medians by $m_l$ and $m_r$, respectively, and their types by $\theta_{m_l}$ and $\theta_{m_r}$.

Let $\mathcal{P} = \{x_p : p \in \mathcal{CP}\}$ be the multiset of proposals (recall the concept of a multiset we introduced in the proof of Proposition \ref{proposition:no_eq_odd_voting}). Define:
\[
x_{\min} = \min \mathcal{P}, \quad x_{\max} = \max \mathcal{P}
\]
Clearly, \(x_{\min} \leq x_{\max}\). Moreover, we cannot have $x_{\min} = x_{\max}$: since $k < N - 1$, there are at least three proposers, so if all proposed the same value, every agent would be indifferent between proposing and voting. This would create incentives to deviate, contradicting equilibrium. Hence, \(x_{\min} < x_{\max}\).

Now consider any two proposals $x_1, x_2 \in \mathcal{P}$. If $x_1$ is strictly farther from both $\theta_{m_l}$ and $\theta_{m_r}$ than $x_2$ is, then $x_1$ will lose a pairwise vote against $x_2$, since $x_2$ would strictly win more than half of the voters.

Therefore, at equilibrium, no proposal can lose against every other one in pairwise comparison—otherwise, the corresponding proposer would strictly prefer to deviate to voting, as their proposal is never implemented and voting can only increase utility.

Similarly, two agents cannot propose identical, losing policies unless those proposals win at least one pairwise comparison. Otherwise, one of them could switch to being a voter without changing the outcome but potentially improving their utility—again contradicting equilibrium.

Thus, every proposal must win at least half the votes against some other proposal. Since $|\mathcal{P}| \geq 3$, there exists a well-defined median proposal, say $x_m$, which is unique because $|\mathcal{P}|$ is odd.

We now rule out the configurations where proposals are too extreme relative to the voters' types:
\[
x_{\min} < x_{\max} < \theta_{m_l} \quad \text{or} \quad \theta_{m_r} < x_{\min} < x_{\max}.
\]
In the first case, $x_{\min}$ would lose against every other proposal (except possibly itself), contradicting the requirement that each proposal wins at least half the votes against some other proposal. The second case leads to a similar contradiction by considering $x_{\max}$.

It remains to analyze the following four configurations:
\begin{enumerate}[label=(\roman*)]
    \item $x_{\min} < \theta_{m_l} < x_{\max} < \theta_{m_r}$
    \item $x_{\min} < \theta_{m_l} < \theta_{m_r} < x_{\max}$
    \item $\theta_{m_l} \leq x_{\min} < x_{\max} \leq \theta_{m_r}$
    \item $\theta_{m_l} < x_{\min} < \theta_{m_r} < x_{\max}$
\end{enumerate}

Case (i): $x_{\min} < \theta_{m_l} < x_{\max} < \theta_{m_r}$: 
We must have $|x_{\max} - \theta_{m_l}| > |x_{\min} - \theta_{m_l}|$; otherwise, $x_{\min}$ becomes the worst proposal and contradicts equilibrium.

Moreover, this configuration is impossible if $|\theta_{m_r} - \theta_{m_l}| < |x_{\min} - \theta_{m_l}|$, as both inequalities cannot hold simultaneously.

Let $\theta_{\min}$ and $\theta_{\max}$ be the types of agents proposing $x_{\min}$ and $x_{\max}$, respectively. We argue that $\theta_{\min} \leq x_{\min}$ (the proof for $\theta_{\max} \geq x_{\max}$ is analogous). If $\theta_{\min} > x_{\min}$, the agent could profitably deviate by increasing their proposal closer to both medians and to their own type, raising both the chance of winning and realized utility.

Assuming without loss of generality that $|x_{\min} - x_m| \leq |x_{\max} - x_m|$, the agent proposing $x_{\min}$ would have a non-decreasing deviation to voting: doing so would make $\theta_{m_l}$ the unique median voter, shifting outcomes toward policies closer to $\theta_{\min}$. Thus, utility increases.

This argument also rules out case (iv) in a symmetric manner.

Case (ii): $x_{\min} < \theta_{m_l} < \theta_{m_r} < x_{\max}$: 
At equilibrium, we must simultaneously have:
\[
|x_{\min} - \theta_{m_l}| \leq |x_{\max} - \theta_{m_l}|, \quad 
|x_{\max} - \theta_{m_r}| \leq |x_{\min} - \theta_{m_r}|,
\]
otherwise a unique worst proposal exists, contradicting equilibrium.

As before, $\theta_{\min} \leq x_{\min}$, since a deviation to increase the proposal would be profitable if $\theta_{\min} > x_{\min}$. Assuming again $|x_{\min} - x_m| \leq |x_{\max} - x_m|$, the proposer of $x_{\min}$ benefits from switching to voting, since $\theta_{m_l}$ becomes the median and the resulting outcomes favor $\theta_{\min}$.

Case (iii): $\theta_{m_l} \leq x_{\min} < x_{\max} \leq \theta_{m_r}$: 
The same reasoning applies: if $\theta_{\min} > x_{\min}$, the agent making the proposal would have a profitable deviation to increase the proposal. Assuming, w.l.o.g. $|x_{\min} - x_m| \leq |x_{\max} - x_m|$, the proposer $\theta_{min}$ again benefits from voting, as the decisive voter becomes $\theta_{m_l}$, improving the expected utility of the agent with type $\theta_{min}$, by having all outcomes closer to his type than before.

In each of the above configurations, a profitable deviation arises for at least one proposer, contradicting the assumption of equilibrium. Thus, no equilibrium with $N - k$ proposers can exist when $k$ is even and $0 < k < N - 1$.
\end{proof}

Now we can finally prove the uniqueness of the equilibrium characterized in Theorem \ref{theorem:vop-theorem} for the case of an odd number of agents.

\begin{proof}[Proof of Theorem \ref{theorem:Uniqueness_of_N_odd}]
    The result follows by combining Proposition \ref{proposition:no_eq_N-1_Proposing}, Proposition \ref{proposition:no_eq_odd_voting}, Proposition \ref{proposition:no_eq_even_voting}, and Theorem \ref{theorem:vop-theorem}.
\end{proof}

\subsection{Example of election for $N=3$ agents}

\subsubsection*{Case 1: Zero Proposers}

If no agent opts to make a proposal, an alternative is selected uniformly at random from the agents' types \(\{\theta_1, \theta_2, \theta_3\}\). The expected utility for each agent \(i\) is:

\[
\mathbb{E}(u_{\theta_i, v}) = -\frac{1}{3} \left( (\theta_i - \theta_1)^2 + (\theta_i - \theta_2)^2 + (\theta_i - \theta_3)^2 \right).
\]

\begin{table}[h]
\centering
\renewcommand{\arraystretch}{1.25} 
\begin{tabular}{|c|c|}
\hline
Agent & Expected Utility \\
\hline
1 & \(-\frac{ (\theta_1 - \theta_2)^2 + (\theta_1 - \theta_3)^2}{3}\) \\
2 & \(-\frac{(\theta_2 - \theta_1)^2 + (\theta_2 - \theta_3)^2}{3}\) \\
3 & \(-\frac{(\theta_3 - \theta_1)^2 + (\theta_3 - \theta_2)^2}{3}\) \\
\hline
\end{tabular}
\caption{Expected utilities when no agent proposes.}
\label{tab:zero_proposers}
\end{table}

In this scenario, the median agent is motivated to deviate and propose himself, as this proposal would result in a utility of \(0\), which is an improvement over the expected utility from random selection. For completeness, it is worth noting that in this strategic profile, every agent has an incentive to self-propose. By being the sole proposer, each agent could achieve a utility of \(0\).

\subsubsection*{Case 2: One Proposer}

Suppose that only one agent opts to make a proposal. The proposer suggests a candidate, and the other two agents vote. The proposer will propose the candidate that maximizes their utility, and hence himself. We repeat Table \ref{tab:one_proposer} from Subsection \ref{IllustrationN3Agents}.

\begin{table}[h]
\centering
\renewcommand{\arraystretch}{1.5}
\begin{tabular}{|c|c|c|c|}
\hline
Proposer & Utility for Agent 1 & Utility for Agent 2 & Utility for Agent 3 \\
\hline
Agent 1 & 0 & \(-(\theta_2 - \theta_1)^2\) & \(-(\theta_3 - \theta_1)^2\) \\
\hline
Agent 2 & \(-(\theta_1 - \theta_2)^2\) & 0 & \(-(\theta_3 - \theta_2)^2\) \\
\hline
Agent 3 & \(-(\theta_1 - \theta_3)^2\) & \(-(\theta_2 - \theta_3)^2\) & 0 \\
\hline
\rowcolor{gray!30}
Equilibrium & No & Yes & No \\
\hline
\end{tabular}
\caption*{Table \ref{tab:one_proposer}: Utilities realized by each agent when only one agent proposes.}
\end{table}
\begin{itemize}
    \item If agent 1 proposes himself, agent 2 has an incentive to propose himself, which would yield to him a higher utility of \(0\). Indeed, if agent 2 proposes, the proposals would be put to a vote, and since agent 3 would choose the proposal closer to his type, agent 2 would win the election since \(\theta_2\) is closer to \(\theta_3\) than \(\theta_1\).

    \item If agent 2 proposes himself, neither agent 1 nor agent 3 has an incentive to deviate. Even if they were to propose, they would not achieve a better outcome than by voting. Specifically, any proposal they make would either result in the same utility or a lower one compared to voting. Given the condition that agents choose to vote if indifferent between voting and proposing, this implies they have no incentive to propose.

    \item If agent 3 proposes \(\theta_3\), agent 2 has an incentive to propose \(\theta_2\), which would yield a higher utility of \(0\).
\end{itemize}

Thus, the only stable scenario is when the median agent (agent 2) proposes himself. This aligns exactly with the equilibrium characterized by Theorem \ref{theorem:vop-theorem}.

\subsubsection*{Case 3: Two Proposers}

We will demonstrate case by case that there is no equilibrium when two agents are proposing. We consider every possible configuration explicitly. Specifically, we examine the following pairs of agents: \((1, 2)\), \((1, 3)\), and \((2, 3)\).

\paragraph{Agents 1 and 2 Proposing}

\begin{itemize}
    \item If agent 1 proposes himself, agent 2, knowing this, would propose himself as he would win the election and achieve maximum utility. However, knowing that agent 2 is proposing himself, agent 1 would then vote, as any proposal he could make would result in lower utility. This shows that no equilibrium exists where agent 1 proposes himself.
    \item If agent 1 proposes agent 2, agent 2 would simply vote, as he is indifferent between voting and proposing. Hence, there is no sustainable equilibrium with agent 1 proposing agent 2.
    \item If agent 1 proposes agent 3, agent 2 would again vote, as he would lose the vote regardless and is indifferent between any proposal he could make and voting.
\end{itemize}

The following table represents all the utilities realized when agents 1 and 2 propose. The blue cells represent the actions that agent 1 can choose, while the orange columns represent the actions available to agent 2. In all white cells, the first entry is the utility of agent 1 in that profile, and the second is the utility of agent 2. We use the following notation:

\[
u_a = -(\theta_1 - \theta_2)^2, \quad u_b = -(\theta_2 - \theta_3)^2, \quad u_c = -(\theta_3 - \theta_1)^2, \quad u_r = \frac{u_a + u_b + u_c}{3}
\]

\begin{table}[h]
\centering
\renewcommand{\arraystretch}{1.5}
\definecolor{lightblue}{RGB}{200, 221, 242}
\definecolor{lightorange}{RGB}{255, 230, 204}

\begin{tabular}{|c|c|c|c|c|}
\hline
Proposals & \cellcolor{lightblue}Agent 1 & \cellcolor{lightblue}Agent 2 & \cellcolor{lightblue}Agent 3 & \cellcolor{lightblue}Vote \\
\hline
\cellcolor{lightorange}Agent 1 & \((0, u_a)\) & \((u_a, 0)\) & \((u_c, u_b)\) & \((0, u_a)\) \\
\hline
\cellcolor{lightorange}Agent 2 & \((u_a, 0)\) & \((u_a, 0)\) & \((u_c, u_b)\) & \((u_a, 0)\) \\
\hline
\cellcolor{lightorange}Agent 3 & \((u_c, u_b)\) & \((u_c, u_b)\) & \((u_c, u_b)\) & \((u_c, u_b)\) \\
\hline
\cellcolor{lightorange}Vote & \((0, u_a)\) & \((u_a, 0)\) & \((u_c, u_b)\) & \((u_r, u_r)\) \\
\hline
\end{tabular}
\caption{Outcomes for every possible combination of proposals when Agent 1 and Agent 2 are proposing.}
\label{tab:2_proposer}
\end{table}

Notice that proposing agent 2 is a dominant strategy for agent 2. Consequently, agent 1 will opt to vote when agent 2 proposes himself, as agent 1 is indifferent between voting and proposing either agent 1 or agent 2. This observation provides another perspective on why no equilibrium with both proposing exists in this scenario.

\paragraph{Agents 2 and 3 Proposing}

The same discussion applies to the scenario where agents 2 and 3 are proposing, leading to the conclusion that no equilibrium exists with these two agents proposing.

\paragraph{Agents 1 and 3 Proposing}

Assume first, without loss of generality, that agent 1 is closer to agent 2 than agent 3.

\begin{center}
\begin{tikzpicture}

\draw[->] (0,0) -- (10,0);

\filldraw (2,0) circle (3pt) node[above]{Agent 1};
\filldraw (4,0) circle (3pt) node[above]{Agent 2};
\filldraw (9,0) circle (3pt) node[above]{Agent 3};

\end{tikzpicture}
\end{center}

\begin{table}[h]
\centering
\renewcommand{\arraystretch}{1.5}
\definecolor{lightblue}{RGB}{200, 221, 242}
\definecolor{lightorange}{RGB}{255, 230, 204}

\begin{tabular}{|c|c|c|c|c|}
\hline
Proposals & \cellcolor{lightblue}Agent 1 & \cellcolor{lightblue}Agent 2 & \cellcolor{lightblue}Agent 3 & \cellcolor{lightblue}Vote \\
\hline
\cellcolor{lightorange}Agent 1 & \((0, u_c)\) & \((u_a, u_b)\) & \((0, u_c)\) & \((0, u_c)\) \\
\hline
\cellcolor{lightorange}Agent 2 & \((u_a, u_b)\) & \((u_a, u_b)\) & \((u_a, u_b)\) & \((u_a, u_b)\) \\
\hline
\cellcolor{lightorange}Agent 3 & \((0, u_c)\) & \((u_a, u_b)\) & \((u_c, 0)\) & \((u_c, 0)\) \\
\hline
\cellcolor{lightorange}Vote & \((0, u_c)\) & \((u_a, u_b)\) & \((u_c, 0)\) & \((u_r, u_r)\) \\
\hline
\end{tabular}
\caption{Outcomes for every possible combination of proposals when Agent 1 and Agent 3 are proposing.}
\label{tab:2_proposerVers2}
\end{table}

Again, there is no equilibrium in this case, as agent 1 proposing himself is a dominant strategy. Hence, he will propose himself. Then, given that agent 1 proposes himself, agent 3 will propose agent 2 as this is the best strategy for him. Given that agent 3 is proposing agent 2, agent 1 is indifferent between every action he can take and will vote due to the condition we impose. Again, there is no equilibrium with both agents proposing.

The case where agent 1 and agent 3 are equally distant from agent 2 can be resolved similarly, with the only difference being that agent 2 will randomly choose one of the two proposals. Actually, to be completely correct, we should say that agent 2 abstains, since he is indifferent between the two, and then one of the two is chosen uniformly at random. However, since both agents have an incentive to deviate and propose \(\theta_2\), this cannot be an equilibrium. This holds true because:

\[
\mathbb{E}[u_{\theta_i, \theta_2}] = -(\theta_i - \theta_2)^2 > -\frac{(\theta_3 - \theta_1)^2}{2} = \mathbb{E}[u_{\theta_i, \theta_i}] \iff 2(\theta_i - \theta_2)^2 < (\theta_3 - \theta_1)^2 = 4(\theta_i - \theta_2)^2.
\]
\begin{center}
\begin{tikzpicture}

\draw[->] (0,0) -- (8,0);

\filldraw (2,0) circle (3pt) node[above]{Agent 1};
\filldraw (4,0) circle (3pt) node[above]{Agent 2};
\filldraw (6,0) circle (3pt) node[above]{Agent 3};

\end{tikzpicture}
\end{center}
We have examined each case step by step and through different explanations to demonstrate why there is no equilibrium with two agents proposing. The core conclusion is that in every situation, there is always an agent who is indifferent between voting and proposing and will thus end up voting.

\subsubsection*{Case 4: Three Proposers}

If all three agents opt to make a proposal, each agent proposes their own type. Then, by the edge case rules, one of the submitted proposals will be selected uniformly at random. The expected utility for each agent \(i\) is:

\[
\mathbb{E}(u_{\theta_i, \theta_i}) = -\frac{1}{3} \left( (\theta_i - \theta_1)^2 + (\theta_i - \theta_2)^2 + (\theta_i - \theta_3)^2 \right).
\]

\begin{table}[h]
\centering
\renewcommand{\arraystretch}{1.5}
\begin{tabular}{|c|c|}
\hline
Agent & Expected Utility \\
\hline
1 & \(-\frac{1}{3}\left( (\theta_1 - \theta_2)^2 + (\theta_1 - \theta_3)^2\right)\) \\
2 & \(-\frac{1}{3}\left((\theta_2 - \theta_1)^2 + (\theta_2 - \theta_3)^2\right)\) \\
3 & \(-\frac{1}{3}\left((\theta_3 - \theta_1)^2 + (\theta_3 - \theta_2)^2 \right)\) \\
\hline
\end{tabular}
\caption{Expected utilities when all agents propose.}
\label{tab:three_proposers}
\end{table}

Without loss of generality, assume that \(|\theta_1 - \theta_2| \leq |\theta_3 - \theta_2|\). In this case, agent 1 has an incentive to deviate and vote, as voting would yield a higher expected utility.

Indeed, if he is the only one voting, he will vote for 2, and since he is the only one voting, this candidate will be elected. We then have:

\[
-\frac{1}{3}\left((\theta_1 - \theta_2)^2 + (\theta_1 - \theta_3)^2\right) \leq -(\theta_1 - \theta_2)^2
\]

This inequality holds if and only if:

\[
2(\theta_1 - \theta_2)^2 \leq (\theta_1 - \theta_3)^2
\]

Given our assumption that \(|\theta_1 - \theta_2| \leq |\theta_3 - \theta_2|\), it follows that:

\[
2(\theta_1 - \theta_2)^2 \leq 4(\theta_1 - \theta_2)^2 \leq (\theta_1 - \theta_3)^2
\]

This confirms the result.

\subsubsection*{Conclusion}

The only equilibrium occurs when the median agent (agent 2) chooses to engage in proposal-making, proposes himself, and is subsequently elected by default. In all other scenarios, at least one agent has an incentive to deviate, thereby confirming the uniqueness of this equilibrium.

\end{document}